\documentclass[onecolumn,accepted=2019-06-26]{quantumarticle}
\pdfoutput=1
\usepackage{graphicx}
\usepackage[pdftex,colorlinks=true,linkcolor=blue,citecolor=blue]{hyperref}
\usepackage{amsmath, amsthm, amssymb}
\usepackage{mathtools}
\usepackage{subfigure}
\usepackage{comment}
\usepackage{url}
\usepackage[ruled,lined,linesnumbered]{algorithm2e}
\usepackage{tikz}
\usepackage[numbers,sort&compress]{natbib}

\newcommand{\R}{\mathbb{R}}

\newcommand{\C}{\mathbb{C}}
\newcommand{\Z}{\mathbb{Z}}

\newcommand{\E}{\,\mathbb{E}}

\newcommand{\ket}[1]{| #1 \rangle}
\newcommand{\bra}[1]{\langle #1|}
\newcommand{\braket}[1]{\langle #1 \rangle}
\newcommand{\hsip}[2]{\langle #1,#2 \rangle}
\newcommand{\proj}[1]{| #1 \rangle \langle #1 |}

\DeclareMathOperator{\tr}{tr}
\DeclareMathOperator{\rank}{rank}

\DeclareMathOperator{\sgn}{sgn}

\renewcommand{\Re}{\operatorname{Re}}

\newcommand{\eee}{\mathrm{e}}

\newcommand{\be}{\begin{equation}}
\newcommand{\ee}{\end{equation}}
\newcommand{\bea}{\begin{eqnarray}}
\newcommand{\eea}{\end{eqnarray}}
\newcommand{\bes}{\begin{equation*}}
\newcommand{\ees}{\end{equation*}}
\newcommand{\beas}{\begin{eqnarray*}}
\newcommand{\eeas}{\end{eqnarray*}}

\makeatletter
\newtheorem*{rep@theorem}{\rep@title}
\newcommand{\newreptheorem}[2]{%
\newenvironment{rep#1}[1]{%
 \def\rep@title{#2 \ref{##1} (restated)}%
 \begin{rep@theorem}}%
 {\end{rep@theorem}}}
\makeatother

\newtheorem{thm}{Theorem}
\newtheorem*{thm*}{Theorem}
\newtheorem{cor}[thm]{Corollary}

\newtheorem{lem}[thm]{Lemma}
\newtheorem*{lem*}{Lemma}
\newtheorem{claim}[thm]{Claim}
\newtheorem{prop}[thm]{Proposition}

\newtheorem{fact}[thm]{Proposition}
\newtheorem{dfn}{Definition}
\newreptheorem{thm}{Theorem}
\newreptheorem{lem}{Lemma}
\newreptheorem{cor}{Corollary}
\newreptheorem{fact}{Proposition}

\begin{document}

\title{Quantum states cannot be transmitted efficiently \mbox{classically}}
\author{Ashley Montanaro}
\affiliation{School of Mathematics, University of Bristol, UK}
\email{ashley.montanaro@bristol.ac.uk}
\maketitle

\begin{abstract}
We show that any classical two-way communication protocol with shared randomness that can approximately simulate the result of applying an arbitrary measurement (held by one party) to a quantum state of $n$ qubits (held by another), up to constant accuracy, must transmit at least $\Omega(2^n)$ bits. This lower bound is optimal and matches the complexity of a simple protocol based on discretisation using an $\epsilon$-net. The proof is based on a lower bound on the classical communication complexity of a distributed variant of the Fourier sampling problem. We obtain two optimal quantum-classical separations as easy corollaries. First, a sampling problem which can be solved with one quantum query to the input, but which requires $\Omega(N)$ classical queries for an input of size $N$. Second, a nonlocal task which can be solved using $n$ Bell pairs, but for which any approximate classical solution must communicate $\Omega(2^n)$ bits.
\end{abstract}


\section{Introduction}
\label{sec:intro}

How much classical information does it take to store or transmit a quantum state? In one sense, the answer is clear: a pure state of $n$ qubits corresponds to a unit vector in $\C^{2^n}$, which requires $2^{n+1}-2$ real numbers to be specified exactly ($2^n$ complex numbers, except that the absolute value of the last one is already determined by normalisation, and we can ignore an irrelevant overall phase). However, surprisingly, a number of known results suggest that the amount of classical information required to transmit a quantum state could actually be substantially less than this.

The most famous result of this nature is Holevo's Theorem~\cite{holevo73}, a corollary of which is that the classical information content of a quantum state of $n$ qubits is bounded by $n$ bits. A related result is a bound of Nayak~\cite{nayak99a} which implies that any quantum communication protocol which transmits $n$ bits with success probability $\delta > 0$ must send at least $n - \log_2 (1/\delta)$ qubits. Indeed, the number of qubits transmitted must be linear in $n$ even if we seek only to retrieve one of the bits with high probability~\cite{ambainis02a}. It was also shown by Aaronson~\cite{aaronson07} that to predict the outcomes of most measurements drawn from some probability distribution on an $n$-qubit state, one needs to make only $O(n)$ sample measurements drawn from the same distribution. This result motivated Aaronson to make the provocative statement that ``while the effective dimension of an $n$-qubit Hilbert space {\em appears} to be exponential in $n$, in the sense that is relevant for approximate learning and prediction, this appearance is illusory''.

One way to gain intuition for these results is to observe that having one copy of an $n$-qubit quantum state $\ket{\psi}$ does not allow the retrieval of up to around $2^n$ parameters precisely. For any measurement which we can perform on $\ket{\psi}$, if we instead applied it to $\ket{\widetilde{\psi}} \approx \ket{\psi}$, the distribution on measurement outcomes would be almost the same. Therefore, the most we can reasonably ask of a classical protocol designed to store or transmit $\ket{\psi}$ is to achieve what the quantum protocol itself does: allow approximate reproduction of any measurement which we could perform on $\ket{\psi}$.

We can express this task within the framework of a communication game. Imagine there are two parties (Alice and Bob), where Alice has a description of a pure quantum state $\ket{\psi}$ of $n$ qubits, and Bob has a description of a quantum measurement (POVM) $M$. Let $p_M(\ket{\psi})$ denote the probability distribution resulting from applying the measurement $M$ to $\ket{\psi}$. Then we ask that the classical protocol allows Alice and Bob to sample from a distribution $\widetilde{p}_M(\ket{\psi})$ such that $\| p_M(\ket{\psi}) - \widetilde{p}_M(\ket{\psi})\|_1 \le \epsilon$, for some $\epsilon > 0$, where $\|\cdot \|_1$ is the $\ell_1$ distance $\|v\|_1 = \sum_i |v_i|$. Call this task {\em Distributed Quantum Sampling} with inaccuracy $\epsilon$. Distributed Quantum Sampling can clearly be solved with $\epsilon=0$ by a quantum protocol communicating $n$ qubits: Alice just sends $\ket{\psi}$ to Bob, who measures according to $M$.

Distributed Quantum Sampling can also be solved with $O(2^n \log (1/\epsilon))$ bits of classical communication. It is sufficient for Alice to encode $\ket{\psi}$ using an $\epsilon$-net with respect to the trace norm~\cite{hayden04}, i.e.\ a set of states $\{ \ket{\psi_i}\}$ such that, for all $\ket{\psi}$, $\| \proj{\psi} - \proj{\psi_i} \|_1 \le \epsilon$ for some $i$. Then Alice sends the identity of the closest state in the $\epsilon$-net to Bob, who samples from the distribution corresponding to applying $M$ to that state. As there are $\epsilon$-nets of size $(5/\epsilon)^{2^{n+1}}$ for the space of pure states of $n$ qubits with respect to the trace norm~\cite{hayden04}, the identity of a state in the net can be transmitted using $O(2^n \log (1/\epsilon))$ bits of communication\footnote{Note that the complexity achieved by this approach is better than would be obtained by simply writing down $\ket{\psi}$ in the computational basis, and truncating each amplitude after some number of digits. To achieve sufficient precision with this approach requires specifying each amplitude up to precision $O(\epsilon / \sqrt{2^n})$, giving an overall communication complexity of $O(2^n \log (2^n/\epsilon))$.}. An $\epsilon$-net argument can also be used to deal with a technicality: to completely specify Alice and Bob's inputs would require an unbounded amount of classical information, leading to their protocols being only describable via a nonstandard computational model. Restricting Alice's input to be picked from an $\epsilon$-net reduces the amount of information required to specify it to $O(2^n)$ bits for constant $\epsilon > 0$, without substantially changing the complexity of the problem (if Distributed Quantum Sampling with inaccuracy $\epsilon$ can be solved for states in an $\epsilon$-net, this implies a protocol for Distributed Quantum Sampling for arbitrary states with inaccuracy $2\epsilon$). A similar $\epsilon$-net idea can be applied to discretise Bob's input, assuming that we restrict to the case where his measurement has a finite number of outcomes, as the set of POVMs on a finite-dimensional space with a fixed finite number of outcomes is compact~\cite{watrous18}.

It is easy to see by a volume argument that the result of~\cite{hayden04} is tight up to constant terms; that is, any $\epsilon$-net must have size at least $(c/\epsilon)^{2^n}$ for some constant $c>0$. But could we do better than this by using a protocol which is not based on simply discretising the space of quantum states via an $\epsilon$-net? Questions of this nature are studied in the field of quantum communication complexity~\cite{buhrman10}. One of the first results in this area was work of Buhrman, Cleve and Wigderson~\cite{buhrman98} which implies that, in our terminology, Distributed Quantum Sampling with $\epsilon=0$ requires $\Omega(2^n)$ bits of classical communication. Their result is based on proving an $\Omega(2^n)$ lower bound on the deterministic communication complexity of a distributed version of the Deutsch-Jozsa problem in quantum query complexity. However, the complexity of this problem drops to $O(1)$ if $\epsilon$ is allowed to be nonzero. 

Following this, a succession of results showed stronger separations between quantum and classical communication complexity. Raz~\cite{raz99} gave a communication task which could be solved by a two-way quantum protocol communicating $O(n)$ qubits, but which requires $\Omega(\sqrt{2^n})$ bits of classical communication. Bar-Yossef, Jayram and Kerenidis~\cite{baryossef08} showed that there is a communication task which can be solved by a {\em one-way} quantum protocol communicating $O(n)$ qubits, while any classical one-way bounded-error protocol must communicate $\Omega(\sqrt{2^n})$ bits. Gavinsky et al.~\cite{gavinsky08a} later proved a similar separation for a functional problem, i.e.\ one where Alice and Bob's task is to compute a function rather than sample from a distribution.

Finally, it was shown by Klartag and Regev that Distributed Quantum Sampling requires $\Omega(\sqrt[3]{2^n})$ classical bits of communication between Alice and Bob~\cite{klartag11}, even if the communication is allowed to be two-way. This improved a previous result of Gavinsky~\cite{gavinsky08b}, which implied that Distributed Quantum Sampling requires $\Omega(\sqrt[8]{2^n}/\sqrt{n})$ bits of classical two-way communication. The lower bound of~\cite{klartag11} is proven by considering a more restrictive problem known as the ``vector in subspace'' problem, which is defined as follows. Alice gets an $n$-qubit quantum state $\ket{\psi}$ and Bob gets a 2-outcome projective measurement $\{M,I-M\}$. Alice and Bob are promised that either $\braket{\psi|M|\psi} = 1$ or $\braket{\psi|M|\psi} = 0$; their task is to determine which is the case. This problem encompasses all one-way exact quantum protocols where Bob has two possible outputs~\cite{kremer95}.

Thus, all these results (and more discussed in Section \ref{sec:prior} below) leave open a natural question: could there exist a non-trivial classical protocol for Distributed Quantum Sampling for fixed $\epsilon>0$, i.e.\ one that transmits asymptotically less than $2^n$ bits? Such a protocol indeed exists for the vector in subspace problem: it was already shown by Raz~\cite{raz99} that this problem can be solved with bounded failure probability using $O(\sqrt{2^n})$ bits of classical one-way communication\footnote{This result was stated in~\cite{raz99} but the proof has not appeared. We include a proof in Appendix~\ref{app:raz}.}. Could the same be true for the more general Distributed Quantum Sampling problem? In the author's opinion, the interest in this question goes beyond simply finding a tight bound for this problem that would improve the best $\Omega(\sqrt[3]{2^n})$ lower bound known. On a fundamental, conceptual level, the question asks: are quantum states ``really'' like an exponentially-long string of numbers, or do they have a more efficient representation?


\subsection{Our results}

Here we show that any classical communication protocol for Distributed Quantum Sampling with sufficiently small constant inaccuracy $\epsilon > 0$ must transmit $\Omega(2^n)$ classical bits, even if the communication is allowed to be two-way and the parties are allowed shared randomness. This immediately implies that any classical method for storing an arbitrary quantum state such that measurements can be approximately simulated on that state must store $\Omega(2^n)$ classical bits. This can be seen as an ``anti-Holevo'' theorem: a quantum state of $n$ qubits can only store at most $n$ bits~\cite{holevo73}, but $\Omega(2^n)$ classical bits are required to store $n$ qubits.

This optimal lower bound is based on proving a quantum-classical separation for the following special case of Distributed Quantum Sampling, which we call Distributed Fourier Sampling:
\begin{itemize}
\item Alice is given a function $f:\{0,1\}^n \rightarrow \{\pm1\}$.
\item Bob is given a function $g:\{0,1\}^n \rightarrow \{\pm1\}$.
\item Their task is for one party (say Bob) to approximately sample from the distribution $p_{fg}$ on $n$-bit strings $s$ where
\be \label{eq:pfg} p_{fg}(s) = \left(\frac{1}{2^n} \sum_{x \in \{0,1\}^n} (-1)^{s \cdot x} f(x) g(x) \right)^2 \ee
and $s \cdot x = \sum_{i=1}^n s_i x_i$. That is, Bob must output a sample from any distribution $\widetilde{p}_{fg}$ such that $\| \widetilde{p}_{fg} - p_{fg} \|_1 \le \epsilon$, for some constant  inaccuracy $\epsilon$.
\end{itemize}
The title ``Distributed Fourier Sampling'' refers to the fact that the distribution which must be sampled from is the square of the Fourier transform of the function $fg(x) = f(x)g(x)$ over $\Z_2^n$.

For conciseness, we henceforth write $N=2^n$. Distributed Fourier Sampling can be solved with $n$ qubits of one-way communication and $\epsilon=0$. Alice constructs the state $\ket{\psi_f} = \frac{1}{\sqrt{N}}\sum_{x \in \{0,1\}^n} f(x) \ket{x}$ and sends it to Bob. Bob then applies the unitary operator defined by $U_g\ket{x} = g(x) \ket{x}$ to $\ket{\psi_f}$ to produce $\ket{\psi_{fg}} = \frac{1}{\sqrt{N}}\sum_{x \in \{0,1\}^n} f(x) g(x) \ket{x}$. Finally, Bob applies a Hadamard gate to each qubit of $\ket{\psi_{fg}}$ and measures in the computational basis. The resulting distribution is exactly $p_{fg}$.

By contrast, we have the following result:

\begin{thm}
\label{thm:main}
There exist universal constants $\epsilon, \gamma > 0$ such that, for sufficiently large $N$, any two-way classical communication protocol for Distributed Fourier Sampling with shared randomness and inaccuracy $\epsilon$ must communicate at least $\gamma N$ bits.
\end{thm}

Theorem \ref{thm:main} implies an optimal lower bound on Distributed Quantum Sampling for $\epsilon = \Omega(1)$, matching the general upper bound of $O(N \log (1/\epsilon))$ bits. It remains open to prove a tight bound when $\epsilon = o(1)$. However, note that a lower bound of $\Omega(N + \log(1/\epsilon))$ does hold, because an $\Omega(\log (1/\epsilon))$ bound is straightforward. If we let Alice's state be $\alpha \ket{0} + \beta \ket{1}$, and Bob's measurement is fixed to be a computational basis measurement, sampling from the corresponding probability distribution up to inaccuracy $\epsilon$ is effectively the same as transmitting $|\alpha|^2$ up to inaccuracy $O(\epsilon)$. This requires transmitting $\Omega(\log (1/\epsilon))$ bits.


\subsection{Consequences of the lower bound}

Theorem \ref{thm:main} immediately implies similar separations in related models.

{\bf Query complexity of sampling problems.} We can use Distributed Fourier Sampling to obtain a lower bound in the query model on the classical complexity of the (non-distributed) Fourier Sampling problem~\cite{bernstein97,aaronson15a}. In this problem, we are given the ability to query (evaluate) an unknown function $h:\{0,1\}^n \rightarrow \{\pm1\}$, which corresponds to an input of size $N$. Our task is to approximately sample from the Fourier spectrum of $h$, i.e.\ the distribution $p_h$ on bit-strings $s \in \{0,1\}^n$ where
\[ p_h(s) = \left(\frac{1}{2^n} \sum_{x \in \{0,1\}^n} (-1)^{s \cdot x} h(x) \right)^2. \]
More precisely, we are asked to output a sample from any distribution $\widetilde{p}_h$ such that $\| \widetilde{p}_h - p_h \|_1 \le \epsilon$. This problem can be solved exactly with 1 quantum query to $h$ by constructing the state $\frac{1}{\sqrt{2^n}} \sum_{x \in \{0,1\}^n} h(x) \ket{x}$, applying a Hadamard gate to each qubit, and measuring in the computational basis.

Any randomised classical query algorithm solving Fourier Sampling using $t$ queries immediately implies a classical two-way communication protocol with shared randomness for Distributed Fourier Sampling communicating at most $2t$ bits, via a standard reduction~\cite{kushilevitz97}. Alice and Bob simulate the procedure for Fourier Sampling, and whenever they want to query $h(x)$, they replace the query with evaluating $f(x)g(x)$ using 2 bits of communication\footnote{Note that, unlike in the setting of quantum query algorithms, it is not necessary for Alice and Bob to send each other their query indices. The classical algorithm can be expressed as a distribution over deterministic decision trees. Alice and Bob choose which deterministic tree to implement using shared randomness, and from that point onwards, the values of the queried bits completely determine the algorithm's behaviour.}. Therefore, Theorem \ref{thm:main} implies a corresponding lower bound on the query complexity of Fourier Sampling:

\begin{cor}
\label{cor:query}
For sufficiently small constant $\epsilon > 0$, any randomised classical algorithm solving Fourier Sampling on $N$ input bits must make $\Omega(N)$ queries.
\end{cor}

A lower bound of $\Omega(N/\log N)$ queries on Fourier Sampling was previously shown by Aaronson and Ambainis~\cite{aaronson15a}, who conjectured (on p44 of the first version of the paper on the arXiv) that this bound was tight, not only for Fourier Sampling, but for {\em all} sampling problems that can be solved with 1 quantum query. Corollary \ref{cor:query} refutes this conjecture. Interestingly, it was also shown by the same authors that any partial boolean function which can be computed by a bounded-error quantum algorithm making $t=O(1)$ queries can be computed by a bounded-error classical algorithm making $O(N^{1-1/(2t)})$ queries~\cite{aaronson15a}. Thus, to see maximal quantum-classical query separations, we are required to go beyond computing boolean functions.

Following the completion of an initial version of this work, I learned that Scott Aaronson and Lijie Chen have recently obtained an independent proof of Corollary \ref{cor:query}~\cite{aaronson16}. (Their proof technique is compared with the present one in Section \ref{sec:sketch} below.) Also note that it is much simpler to prove Corollary \ref{cor:query} for deterministic classical algorithms, i.e.\ ones that choose which queries to make via a deterministic decision tree, then sample from some distribution depending on the results of the queries. A proof for this special case is included in Appendix~\ref{app:query}.

{\bf Nonlocality problems.} Using a standard mapping between communication protocols and entanglement~\cite[Section IV]{buhrman10}, we can obtain a distribution $\mathcal{D}$ which can be sampled from exactly with no communication between the parties if Alice and Bob share $n$ Bell pairs, but such that any classical procedure for sampling from $\mathcal{D}$ up to distance $\epsilon$ in $\ell_1$ norm, for some constant $\epsilon > 0$, requires $\Omega(2^n)$ bits of classical communication. A similar lower bound for {\em exact} sampling from the same distribution $\mathcal{D}$ was previously shown by Brassard, Cleve and Tapp~\cite{brassard99}, but the bounded-error case was called ``an important open question'' by Toner and Bacon~\cite{toner03}.

To obtain $\mathcal{D}$, we define a variant of the Distributed Fourier Sampling problem. Alice and Bob are again each given a function, $f$ and $g$ respectively. However, this time they are asked to sample from a distribution on pairs of bit-strings $s,t \in \{0,1\}^n$ (where Alice outputs $s$, Bob outputs $t$) of the following form, up to $\ell_1$ distance $\epsilon$: a distribution where $\Pr[s \oplus t = u] = p_{fg}(u)$ for all $u \in \{0,1\}^n$, where $p_{fg}$ is defined as in (\ref{eq:pfg}). We call this problem Doubly Distributed Fourier Sampling (DDFS).

The quantum protocol for DDFS proceeds as follows. Alice and Bob share a maximally entangled state $\frac{1}{\sqrt{N}} \sum_{x \in \{0,1\}^n} \ket{x}\ket{x}$. Alice and Bob each apply the unitary operators $U_f$ and $U_g$, defined by $U_f\ket{x} = f(x) \ket{x}$, $U_g\ket{x} = g(x) \ket{x}$, to their half of the state to produce the state $\ket{\phi_{fg}} = \frac{1}{\sqrt{N}} \sum_{x \in \{0,1\}^n} f(x) g(x) \ket{x}\ket{x}$. They each then apply Hadamard gates to each qubit of their half of $\ket{\phi_{fg}}$, measure in the computational basis, and output the result.

The final state produced before measuring is
\[ \frac{1}{N^{3/2}} \sum_{x \in \{0,1\}^n} f(x)g(x) \left(\sum_{s \in \{0,1\}^n} (-1)^{s \cdot x} \ket{s} \right) \left(\sum_{t \in \{0,1\}^n} (-1)^{t \cdot x} \ket{t} \right), \]
so for each pair of bit-strings $(s,t)$ such that $s\oplus t = u$, the probability that Alice and Bob see that pair of bit-strings is
\[ \frac{1}{N^3} \left( \sum_{x \in \{0,1\}^n} (-1)^{x \cdot (s \oplus t)} f(x)g(x) \right)^2 = \frac{p_{fg}(u)}{N}. \]
Therefore, the probability that they see some pair of bit-strings $(s,t)$ such that $s \oplus t = u$ is exactly $p_{fg}(u)$. On the other hand, any classical communication protocol for DDFS approximating the output distribution up to $\ell_1$ inaccuracy $\epsilon$ with cost $c$ gives a protocol for Distributed Fourier Sampling up to inaccuracy $\epsilon$ with cost $c + n$. This protocol proceeds as follows: after receiving outcomes $(s,t)$ from their DDFS protocol, Alice sends $s$ to Bob, who outputs $s \oplus t$. If the probability that Alice and Bob output $(s,t)$ was $\widetilde{q}_{fg}(s,t)$, then as the DDFS protocol achieved $\ell_1$ distance at most $\epsilon$ from a distribution where $\Pr[s \oplus t = u] = p_{fg}(u)$ for all $u \in \{0,1\}^n$, we have
\[ \sum_{s,t \in \{0,1\}^n} | q_{fg}(s,t) - \widetilde{q}_{fg}(s,t) | \le \epsilon \]
for some distribution $q_{fg}$ such that for all $u \in \{0,1\}^n$, $\sum_{s,t: s \oplus t=u} q_{fg}(s,t) = p_{fg}(u)$. So the $\ell_1$ distance between the output distribution and the desired distribution is
\beas \sum_{u \in \{0,1\}^n} \left| p_{fg}(u) - \sum_{s,t:s\oplus t = u} \widetilde{q}_{fg}(s,t) \right| &=& \sum_{u \in \{0,1\}^n} \left| \sum_{s,t:s \oplus t=u} \frac{q_{fg}(s,t)}{p_{fg}(u)} p_{fg}(u) - \widetilde{q}_{fg}(s,t) \right|\\
&\le& \sum_{u \in \{0,1\}^n} \sum_{s,t:s \oplus t=u} | q_{fg}(s,t) - \widetilde{q}_{fg}(s,t) |\\
&=& \sum_{s,t: s,t \in \{0,1\}^n} | q_{fg}(s,t) - \widetilde{q}_{fg}(s,t) | \le \epsilon.
\eeas
We have obtained the following corollary:

\begin{cor}
\label{cor:ddfs}
For sufficiently small constant $\epsilon > 0$, any two-way classical communication protocol with shared randomness for Doubly Distributed Fourier Sampling must communicate $\Omega(N)$ bits.
\end{cor}

It is clear that the bound stated in Corollary \ref{cor:ddfs} is tight up to constant factors, as Alice can send her Doubly Distributed Fourier Sampling input to Bob using $O(N)$ bits of communication. This bound is also best possible for the following more general problem: Alice and Bob initially share an arbitrary bipartite entangled state where each party has local dimension $N$, are each given as input a local measurement of rank-1 projectors, and are asked to sample from the joint distribution on measurement outcomes up to constant inaccuracy $\epsilon > 0$. This task can be achieved classically using $O(N)$ bits of communication, via a similar $\epsilon$-net construction as for Distributed Quantum Sampling. Alice simulates her measurement on her half of the state, after which she knows Bob's reduced state. She then uses a shared $\epsilon$-net to represent this state up to accuracy $\epsilon$, and sends the identity of the state in the $\epsilon$-net to Bob, who can then simulate his own measurement. For constant $\epsilon > 0$, it is sufficient to transmit $O(N)$ bits.

Finally, we remark that the results given here on Fourier Sampling and its distributed variant are connected to some of the earliest works on quantum computation: the first exponential separation between exact quantum and classical algorithms via the Deutsch-Jozsa algorithm~\cite{deutsch92}, the first super-polynomial separation between quantum and randomised classical algorithms via recursive Fourier sampling~\cite{bernstein97}, and the first exponential separations between exact quantum and classical communication complexity via the distributed version of the Deutsch-Jozsa algorithm~\cite{buhrman98,brassard99}. It is remarkable that, around 20 years after these pioneering results, Fourier Sampling continues to be a rich vein from which quantum-classical separations can be mined.


\subsection{Sketch of the proof of the main result}
\label{sec:sketch}

In the remainder of the paper, we prove the claimed lower bound on the classical communication complexity of Distributed Fourier Sampling. We first give an outline of the proof. The starting point is to observe that, if there is a classical protocol for approximately sampling from $p_{fg}$ up to constant inaccuracy $\epsilon$, there is a classical protocol with two outcomes (accept or reject) which communicates the same number of bits and accepts with probability very close to
\be \label{eq:accprot} \left( \frac{\braket{f,g}}{N}\right) ^2 \coloneqq \left(\frac{1}{N} \sum_{x \in \{0,1\}^n} f(x) g(x) \right)^2 \ee
for most inputs $f$, $g$. More specifically, there is a protocol which accepts with probability $\widetilde{p}(f,g)$ such that the average of $|\widetilde{p}(f,g) - (\braket{f,g}/N)^2|$ over uniformly random $f$ and $g$ is at most $\epsilon/N$. A similar idea was used in~\cite{aaronson15a,aaronson16} in the setting of query complexity. We present the reduction in the proof of Lemma \ref{lem:prototoaccprob}.

Next, we show that for sufficiently small $\epsilon > 0$, any classical protocol whose acceptance probability is this close to $(\braket{f,g}/N)^2$ on average over $f$ and $g$ must communicate $\Omega(N)$ bits. It is convenient to now switch notation and consider the problem of accepting with probability close to $(\braket{x,y}/N)^2$ for strings $x,y \in \{\pm1\}^N$. The first challenge in proving a lower bound is that it is hard to reason about protocols for sampling problems. In particular, one cannot necessarily assume that the existence of a bounded-error randomised protocol for a sampling problem implies the existence of a deterministic protocol for the same problem that is accurate on most inputs\footnote{A previous version of this paper claimed incorrectly that such an implication did hold.}. This is by contrast with the case of communication complexity of functional (or relational) problems~\cite{kushilevitz97}, where this connection is standard and is a typical first step in the proof of lower bounds.

We instead need to reason directly about the acceptance probability of any low-communication protocol $P$. Our high-level strategy is to find three distributions $\mathcal{D}_-$, $\mathcal{D}_0$, $\mathcal{D}_+$ on inputs such that we can upper-bound the acceptance probability $p_0$ of $P$ under $\mathcal{D}_0$ in terms of a weighted sum of its acceptance probabilities $p_-$, $p_+$ under $\mathcal{D}_-$, $\mathcal{D}_+$. Assuming that $P$ is accurate allows us to find an upper bound on each of these two acceptance probabilities (and hence an upper bound on $p_0$), as well as a lower bound on $p_0$. If the upper bound is lower than the lower bound, we have obtained a contradiction and there hence cannot exist an accurate low-communication protocol.

To implement this strategy mathematically, we use the fact that a communication protocol can be understood in terms of rectangles, i.e.\ subsets of Alice and Bob's inputs of the form $R = A \times B$. First, it is not difficult to see that the only rectangles that contribute appreciably to the protocol's average acceptance probability under some distribution are those that are large under that distribution. Next we want to show that, for any large rectangle that contributes substantially to the acceptance probability under some distribution $\mathcal{D}_0$, that rectangle must contribute even more substantially to the acceptance probabilities under some other distributions $\mathcal{D}_-$, $\mathcal{D}_+$.

At this point, we can make a connection to a remarkably powerful previous technique used to determine the communication complexity of the so-called gap-Hamming and gap-orthogonality problems~\cite{chakrabarti12,chakrabarti12a,vidick12,sherstov12}. In the gap-Hamming problem Alice and Bob are each given strings $x,y\in \{\pm1\}^N$, under the promise that either $\braket{x,y} \le -\sqrt{N}$ or $\braket{x,y} \ge \sqrt{N}$, where as above $\braket{x,y} = \sum_i x_i y_i$; their goal is to determine which is the case. It was first shown by Chakrabarti and Regev that the communication complexity of the gap-Hamming problem is $\Omega(N)$~\cite{chakrabarti12}. The key technical step of their bound involved giving three distributions $\xi_0$, $\xi_-$, $\xi_+$ on pairs of strings $x,y \in \{\pm1\}^N$ such that, for any large rectangle $R$,
\be \label{eq:crsimple} \frac{1}{2} \left( \xi_-(R) + \xi_+(R) \right) \ge \frac{2}{3} \xi_0(R). \ee
The distributions $\xi_-$ and $\xi_+$ are concentrated on strings with relatively large and small (respectively) Hamming distance; indeed, under these distributions $x$ is uniformly random, while for each $i$, $y_i = x_i$ with probability $(1 \pm p)/2$, for some small $p$. $\xi_0$ is just the uniform distribution. Thus (\ref{eq:crsimple}) can be seen as an anticoncentration bound: any large rectangle (i.e.\ any rectangle that contains many pairs $(x,y)$ with respect to the uniform distribution) must contain many pairs $(x,y)$ such that $|\braket{x,y}|$ is large.

This appears to give us exactly what we need to complete the proof. However, there are some difficulties left to surmount. First, the corresponding acceptance probabilities of a good sampler under the distributions $\xi_0$, $\xi_-$, $\xi_+$ used in~\cite{chakrabarti12} are not the ones that we need to use to prove the desired contradiction. To address this, Alice and Bob randomly flip some of their input bits in a suitably correlated way, which corresponds to shifting the expected acceptance probability $(\braket{x,y}/N)^2$ under each of the distributions. Even after doing this, the bound (\ref{eq:crsimple}) turns out to not be enough to imply a contradiction. We instead need a skewed version of the bound:
\be \label{eq:crsimpleskewed} \frac{1}{2} \left( \eee^s \xi_-(R) + \eee^{-s} \xi_+(R) \right) \ge \frac{2}{3} \xi_0(R) \ee
for large rectangles $R$ and arbitrary $s \in \R$. Intuitively, in this expression we are improving the upper bound on $\xi_0$ in terms of $\xi_+$, at the expense of making the bound in terms of $\xi_-$ worse. It turns out to indeed be possible to prove such a generalised bound based on the techniques of~\cite{chakrabarti12}\footnote{For experts, we remark that the ultimate reason for this is that the bound of~\cite{chakrabarti12} hinges on the behaviour of the function $\cosh(z)$, while here we need to understand the function $\cosh(z-s)$, which is very similar; and that it might also be possible to use a suitably generalised and tightened version of bounds proven in~\cite{chakrabarti12a} to obtain a similar result.}. This generalised bound could find applications elsewhere.

The final issue to deal with is apparently a technicality, but perhaps an important technicality. We have from the accuracy constraint on the Distributed Fourier Sampling protocol that the expected difference between the protocol's acceptance probability and $(\braket{x,y}/N)^2$ over uniformly random $x$ and $y$ is at most $\epsilon/N$. However, for the above bounds to be meaningful we need the protocol to also have a comparable level of accuracy when $(x,y)$ are picked from the distributions $\xi_+$, $\xi_-$. But the distributions $\xi_+$, $\xi_-$ we use are relatively far from the uniform distribution in total variation distance, implying that an expected error of $\epsilon/N$ under the distribution $\xi_0$ may not translate into a small error under the distributions $\xi_+$, $\xi_-$. To address this, Alice and Bob map their inputs $(x,y)$ to new inputs $(x',y')$ that are larger by a constant factor but preserve the inner product $\braket{x,y}$. This allows us to show that all probabilities in the corresponding distributions $\xi'_{\pm}$ on the larger inputs are at most a constant multiple (depending on $\epsilon$) of their corresponding probabilities under the uniform distribution, so the accuracy bound can be translated across.

We finally remark on connections between the proof techniques of the tight bound given here on the communication complexity of Distributed Fourier Sampling, and the tight bound of Aaronson and Chen on the query complexity of Fourier Sampling~\cite{aaronson16}. At a high level, Aaronson and Chen's strategy is also to prove hardness of accepting with probability close to $(\braket{x,y}/N)^2$, based on showing contradictory bounds on the acceptance probability of any classical algorithm that makes few queries, under three different input distributions. As query algorithms are much more restrictive than communication protocols, this allows the acceptance probability of an algorithm in this setting to be written as a weighted sum of certain functions of binomial coefficients. To complete the proof strategy of~\cite{aaronson16} it is then sufficient to carry out some technical calculations to bound these coefficients. In the communication complexity context, it seems necessary to prove a more general bound of the form of (\ref{eq:crsimpleskewed}), which comes with significant technical complications. For example, it does not seem obvious how to prove a variant of (\ref{eq:crsimpleskewed}) corresponding to the distributions used in~\cite{aaronson16}.


\subsection{Barriers to use of previous results}

It is instructive to check why we cannot just use existing communication complexity results to easily prove a lower bound on Distributed Fourier Sampling. First, by the $O(\sqrt{2^n})$ classical upper bound of Raz~\cite{raz99} for any partial boolean function, we can see that we need to look beyond protocols whose acceptance and rejection probabilities are separated by an additive constant. Looking at (\ref{eq:accprot}), one might naturally guess that a tight lower bound could be proven by considering inputs obeying the promise that (for example) either $\hsip{f}{g} = 0$, or $|\hsip{f}{g}| \ge 3\epsilon\sqrt{N}$. In the former case, the protocol should accept with probability at most $\epsilon/N$, and in the latter case the protocol should accept with probability at least $(3\epsilon-\epsilon)/N = 2\epsilon/N$. So there is a multiplicative constant separating the acceptance probabilities in the two cases.

This is a variant of the gap-orthogonality problem~\cite{sherstov12,chakrabarti12a}, for which an $\Omega(N)$ lower bound is known in the case where acceptance and rejection probabilities are separated by an {\em additive} constant. So extending this result to hold in the nonstandard setting where acceptance and rejection probabilities are small and separated by a multiplicative constant would suffice to prove our desired result. Further, G\"o\"os and Watson~\cite{goos14} showed that communication lower bounds of this form can be proven using the corruption method~\cite{yao83,beame07,klauck03}, an important lower bound technique in communication complexity, used in particular by Sherstov to prove his $\Omega(N)$ lower bound for gap-orthogonality~\cite{sherstov12}. If we had a direct corruption bound of $\Omega(N)$ for the gap-orthogonality problem as stated above, this would imply the result we need.

But in fact such a bound cannot exist, because G\"o\"os and Watson also showed that the corruption bound is {\em equivalent} to the communication complexity of an optimal protocol whose acceptance and rejection probabilities can be arbitrarily small, but are separated by a multiplicative constant~\cite{goos14}. And there is indeed a nontrivial protocol of this form for gap-orthogonality: query $\sqrt{N}$ random bits of Alice and Bob's inputs, and accept if they are all different, or all the same. If Alice and Bob's strings differ at a $1/2 + \Delta/\sqrt{N}$ fraction of positions, the probability of acceptance is precisely
\[ \left(\frac{1}{2} + \frac{\Delta}{\sqrt{N}} \right)^{\sqrt{N}} + \left(\frac{1}{2} - \frac{\Delta}{\sqrt{N}} \right)^{\sqrt{N}} = \frac{1}{2^{\sqrt{N}}} \left(\left(1 + \frac{2\Delta}{\sqrt{N}} \right)^{\sqrt{N}} + \left(1 - \frac{2\Delta}{\sqrt{N}} \right) ^{\sqrt{N}} \right). \]
This is roughly equal to $2^{-\sqrt{N}}(\eee^{2\Delta} + \eee^{-2\Delta}) = 2^{1-\sqrt{N}} \cosh(2\Delta)$. For $\Delta$ values separated by a constant factor, the acceptance probabilities will also be separated by a constant factor (which can be made arbitrarily large by taking the AND of multiple runs).

It therefore does not seem possible to use gap-orthogonality to prove the desired lower bound via standard techniques. So how did Sherstov prove a corruption bound of $\Omega(N)$ for gap-orthogonality? He considered the negation of the problem we consider here (i.e.\ reversing the roles of acceptance and rejection), which {\em does} have such a lower bound. However, it is not clear how to use this problem to prove a bound on the original Fourier sampling task, as the acceptance probabilities of the quantum protocol do not correspond directly to the negated problem.


\subsection{Other related prior work}
\label{sec:prior}

The Distributed Quantum Sampling problem has also been studied in the physics literature, where it is sometimes termed ``classical teleportation''. Toner and Bacon~\cite{toner03} showed that, in the case where Bob makes a projective measurement, Distributed Quantum Sampling on one qubit can be solved exactly with 2 bits of communication from Alice to Bob (and shared randomness). An asymptotic protocol which encompasses POVMs and which uses slightly more communication was previously proposed by Cerf, Gisin and Massar~\cite{cerf00a}. Montina~\cite{montina13} gave an efficient classical protocol for the special case where Bob's measurement consists of two operators: the projector onto a pure state, and its complement (a similar result can be obtained from work of Kremer, Nisan and Ron~\cite{kremer99}). Montina, Pfaffhauser and Wolf~\cite{montina13a}, and Montina and Wolf~\cite{montina14}, related the asymptotic and one-shot communication complexities of exact classical simulation of general quantum channels to convex optimisation problems. The related problem of simulating the correlations obtained by measuring entangled states has also been studied by a number of authors; see~\cite{brunner14,buhrman10} for surveys.

Galv\~ao and Hardy~\cite{galvao03} considered a variant of the Distributed Quantum Sampling problem, where Alice starts with a qubit in the state $\ket{0}$, the qubit is sent to Bob via a communication channel in which it undergoes a number of small rotations, and Bob is then asked to determine whether the final state of the qubit is $\ket{0}$ or $\ket{1}$, given that one of these is promised to be the case. Galv\~ao and Hardy argued that any perfect classical simulation of this quantum protocol must use a system which can have infinitely many states, hence must transmit infinitely many bits of information. (See~\cite{brierley15} for a related result applied to the study of non-classical temporal correlations.) However, this lower bound does not hold for approximate simulations (e.g.\ the simulation method based on $\epsilon$-nets mentioned above); in addition, it only holds when arbitrarily many operations can affect the qubit during its progress from Alice to Bob.

We finally remark that, using a connection between the corruption bound and Bell inequalities, the gap-orthogonality problem was recently used by Laplante et al.~\cite{laplante16} to give inefficiency-resistant Bell inequalities with large violations. 


\section{Preliminary definitions}

We use standard definitions from communication complexity~\cite{kushilevitz97}. We consider the setting of two-way communication protocols with shared randomness, with two separated parties (Alice and Bob). In these protocols, Alice and Bob each receive an input (where Alice's input is often denoted $x \in X$, and Bob's input is denoted $y \in Y$) and share a random bit-string of arbitrary length. They then follow a protocol to exchange messages, and at the end of the protocol, one party (e.g.\ Bob) is asked to output a bit-string satisfying some predetermined criteria. The cost of the protocol is the total number of bits communicated. (See~\cite{kushilevitz97} for formal definitions.) We usually consider sampling problems, where Alice and Bob are given the pair $(x,y)$ and asked to output a sample from some predetermined distribution $\mathcal{D}_{xy}$, up to some predetermined accuracy. A {\em rectangle} $R = A \times B$ is a product subset of Alice and Bob's inputs ($A \subseteq X$, $B \subseteq Y$).

$[a,b]$ denotes the interval $\{z \in \R: a \le z \le b\}$. We will use the following family of distributions on correlated pairs $(x,y)$ throughout:

\begin{dfn}[cf.~\cite{chakrabarti12}]
\label{dfn:xi}
For $p \in [-1,1]$, let $\xi^N_p$ denote the distribution on $(x,y) \in \{\pm1\}^N \times \{\pm1\}^N$ obtained by picking $x \in \{\pm1\}^N$ uniformly at random, then picking $y$ by setting $y_i=x_i$ with probability $(1+p)/2$, and $y_i = -x_i$ with probability $(1-p)/2$. When $N$ is implicit, we often write $\xi_p \coloneqq \xi^N_p$.
\end{dfn}

Note that $\xi^N_0$ is the uniform distribution on $\{\pm1\}^N \times \{\pm1\}^N$. 


\section{Proof of main result}

We now describe the proof of Theorem \ref{thm:main} in detail. We first state the technical lemmas from which the main theorem follows, and give its proof assuming these lemmas. The lemmas themselves are proved afterwards.

The first lemma is a characterisation of shared-randomness classical communication protocols in terms of rectangles. Note that similar bounds are common in the communication complexity literature (see e.g.~\cite{jain10b}), but are usually stated only for decision problems, where one wishes to either accept or reject with high probability. As we are interested in the more general setting of sampling problems, we include a proof in Section~\ref{sec:ccproofs}.

\begin{lem}[Characterisation of communication protocols]
\label{lem:proto}
Let $P$ be a communication protocol with two outcomes (accept and reject) and shared randomness, such that $P$ communicates at most $c$ bits. Then there exists a non-negative function $\rho(R)$ of rectangles $R$ such, for any distribution $\mu$ on pairs $(x,y) \in X \times Y$,
\[ \Pr [P \text{ accepts } (x,y)] = \sum_R \rho(R) \mu(R) = \eta + \sum_{R: \mu(R) \ge 2^{-2c}} \rho(R) \mu(R) \]
for some $\eta \in [0,2^{-c})$, where the probability is taken over both $(x,y) \sim \mu$ and the shared randomness of $P$.
\end{lem}

\begin{lem}[From Fourier sampling to two-outcome protocols]
\label{lem:prototoaccprob}
Assume that, for all sufficiently large $N$, there exists a protocol for Distributed Fourier Sampling with inaccuracy $\epsilon$ that communicates at most $\gamma N$ bits. Let $b$ be any positive constant. Then there exists a protocol that communicates at most $2\gamma N$ bits and accepts with probability $\widetilde{p}_{xy}$ on input $(x,y) \in \{\pm1\}^N \times \{\pm1\}^N$ such that:
\[ \E_{(x,y) \sim \xi_{-\sqrt{b/N}}}\left[ \widetilde{p}_{xy} \right] \in \left[\frac{1-\epsilon'}{4N},\frac{1+\epsilon'}{4N}\right], \;\;\;\; \E_{(x,y) \sim \xi_0}\left[ \widetilde{p}_{xy} \right] \in \left[\frac{(b + 1)\beta -\epsilon'}{4N},\frac{(b + 1)\beta+\epsilon'}{4N}\right], \]
\[ \E_{(x,y) \sim \xi_{\sqrt{b/N}}}\left[ \widetilde{p}_{xy} \right] \in \left[\frac{(4b + 1)\beta-\epsilon'}{4N},\frac{(4b + 1)\beta+\epsilon'}{4N}\right], \]
where $\epsilon' = C\eee^{18b} \epsilon$ for some universal constant $C$, and $\beta \in [1-o(1),1]$.
\end{lem}

\begin{lem}[Skewed anticoncentration bound for large rectangles]
\label{lem:tech}
For all $b > 0$ there exists $\delta > 0$ such that for large enough $N$, all rectangles $R \subseteq \{\pm1\}^N \times \{\pm1\}^N$ such that $\xi_0(R) \ge 2^{-\delta N}$, and all~$s$,
\[ \frac{1}{2} \left( \eee^{s} \xi_{-\sqrt{b/N}}(R) + \eee^{-s} \xi_{\sqrt{b/N}}(R) \right) \ge \frac{2}{3} \xi_0(R). \]
\end{lem}

\begin{repthm}{thm:main}
There exist universal constants $\epsilon, \gamma > 0$ such that, for sufficiently large $N$, any two-way classical communication protocol for Distributed Fourier Sampling with shared randomness and inaccuracy $\epsilon$ must communicate at least $\gamma N$ bits.
\end{repthm}

\begin{proof}
Assume towards a contradiction that for all sufficiently large $N$ there exists a protocol for Distributed Fourier Sampling with inaccuracy $\epsilon$ that communicates at most $\gamma N$ bits, for some small constants $\epsilon, \gamma > 0$ which will be chosen later. Let $P$ be the corresponding two-outcome communication protocol that communicates $c \le 2 \gamma N$ bits and has acceptance probability bounds given by Lemma \ref{lem:prototoaccprob}, for some choice of $b = O(1)$ to be determined. By Lemma \ref{lem:proto}, $P$ in turn implies the existence of a non-negative function $\rho(R)$ of rectangles $R$ such that
\be \label{eq:p0} p_0 \coloneqq \Pr_{(x,y) \sim \xi_0} [P \text{ accepts } (x,y)] = \eta + \sum_{R, \xi_0(R) \ge 2^{-2c}} \rho(R) \xi_0(R) \ee
for some $\eta \in [0,2^{-c})$. By Lemma \ref{lem:prototoaccprob},
\[ p_0 \ge \frac{(b + 1)\beta -\epsilon'}{4N}, \]
where $\epsilon' = C \eee^{6b} \epsilon$ and $\beta = 1-o(1)$. By (\ref{eq:p0}) and Lemma \ref{lem:tech}, assuming that $\gamma \le \delta/4$, for any $s \in \R$
\beas
p_0 &\le& 2^{-c} + \sum_{R,\xi_0(R) \ge 2^{-2c}} \rho(R) \frac{3}{4} \left(  \eee^{s} \xi_{-\sqrt{b/N}}(R) + \eee^{-s} \xi_{\sqrt{b/N}}(R) \right) \\
&\le& 2^{-c} + \frac{3}{4} \left( \eee^s \sum_R \rho(R)\xi_{-\sqrt{b/N}}(R) + \eee^{-s} \sum_{R'} \rho(R')\xi_{\sqrt{b/N}}(R')   \right)\\
&\le& 2^{-c} + \frac{3}{4} \left( \eee^s \Pr_{(x,y) \sim \xi_{-\sqrt{b/N}}} [P \text{ accepts } (x,y)] + \eee^{-s} \Pr_{(x,y) \sim \xi_{\sqrt{b/N}}} [P \text{ accepts } (x,y)]  \right) \\
&\le& 2^{-c} + \frac{3}{4} \left(\eee^s \frac{1+\epsilon'}{4N} + \eee^{-s} \frac{(4b+1)\beta+\epsilon'}{4N} \right),
\eeas
where the final inequality is Lemma \ref{lem:prototoaccprob}. Choosing, for example, $b = 10$, $s = 2$, we get that for some universal constant $D$,
\[ \frac{11}{4N} \le \frac{3}{16N}\left(\eee^2 + \frac{41}{\eee^2} \right) + D \frac{\epsilon'}{N}  + o(1/N) \le \frac{10}{4N} + D \frac{\epsilon'}{N}  + o(1/N), \]
where we have folded the $2^{-c}$ term into the $o(1/N)$ term, which is justified by assuming that $P$ communicates exactly $2\gamma N$ bits, up to rounding to an integer (if there exists a protocol that solves a problem by communicating at most $k$ bits for some $k$, there exists a protocol that solves the same problem by communicating exactly $k$ bits). This is a contradiction for small enough $\epsilon = \Omega(1)$ and large enough $N$. Therefore such a protocol cannot exist.
\end{proof}


\section{Proofs of communication complexity lemmas}
\label{sec:ccproofs}

We first prove a mathematical characterisation of communication protocols in terms of rectangles.

\begin{replem}{lem:proto}
Let $P$ be a communication protocol with two outcomes (accept and reject) and shared randomness, such that $P$ communicates at most $c$ bits. Then there exists a non-negative function $\rho(R)$ of rectangles $R$ such, for any distribution $\mu$ on pairs $(x,y) \in X \times Y$,
\[ \Pr [P \text{ accepts } (x,y)] = \sum_R \rho(R) \mu(R) = \eta + \sum_{R, \mu(R) \ge 2^{-2c}} \rho(R) \mu(R) \]
for some $\eta \in [0,2^{-c})$, where the probability is taken over both $(x,y) \sim \mu$ and the internal randomness of $P$.
\end{replem}

\begin{proof}
For any $(x,y)$, we have
\[ \Pr[P \text{ accepts } (x,y)] = \sum_D \pi(D) [D \text{ accepts } (x,y)] = \sum_D \pi(D) \sum_{R \in \mathcal{R}_D} [(x,y) \in R], \]
where $\pi(D)$ is a distribution on deterministic protocols $D$, each of which corresponds to a disjoint set of 1-rectangles $\mathcal{R}_D$. Taking the expectation over $\mu$,
\beas
\Pr_{(x,y) \sim \mu,D \sim \pi} [P \text{ accepts } (x,y)] &=& \sum_D \pi(D) \sum_{R \in \mathcal{R}_D} \mu(R)\\
&=& \sum_D \pi(D) \sum_{R \in \mathcal{R}_D,\mu(R) < 2^{-2c} } \mu(R) + \sum_D \pi(D) \sum_{R \in \mathcal{R}_D,\mu(R) \ge 2^{-2c} } \mu(R)\\
&=& \eta + \sum_D \pi(D) \sum_{R \in \mathcal{R}_D,\mu(R) \ge 2^{-2c} } \mu(R),
\eeas
where we infer that $0 \le \eta < 2^{-c}$ by using $|\mathcal{R}_D| \le 2^c$ for all $D$ and $\sum_D \pi(D) = 1$. The claim follows by taking $\rho(R) = \sum_{D,R \in \mathcal{R}_D} \pi(D)$.
\end{proof}

Next we prove a lemma that maps between a distributed Fourier sampler and a two-outcome communication protocol whose acceptance probability obeys certain bounds. In order to prove the lemma, we need the following technical claims, whose proofs are deferred to the very end (Section \ref{sec:finalproofs}):

\begin{fact}
\label{claim:avgprob}
\[ \E_{(x,y) \sim \xi^N_p}\left[ \left( \frac{\braket{x,y}}{N} \right)^2 \right] =  \frac{1}{N} + \left(1- \frac{1}{N}\right) p^2. \]
\end{fact}

\begin{fact}
\label{fact:techbound}
Let $\xi'_p$ be the distribution on pairs $x',y' \in \{\pm1\}^{2N}$ formed by choosing $(x,y) \sim \xi_p$, then setting $x' = \sigma(x,x'')$, $y' = \sigma(y,y'')$, where $x'',y'' \in \{\pm1\}^N$ are fixed strings such that $\braket{x'',y''} = 0$, and $\sigma \in S_{2N}$ is a random permutation of the indices of the strings. Then there exists a universal constant $C$ such that, for sufficiently large $N$ and any $\Delta$, and $p \in [-0.01,0.01]$,
\[ \Pr_{(x',y')\sim \xi'_p}[\braket{x',y'} = \Delta] \le C \eee^{2p^2 N} \Pr_{(x',y') \sim \xi^{2N}_0}[\braket{x',y'} = \Delta]. \]
\end{fact}

\begin{replem}{lem:prototoaccprob}
Assume that, for all sufficiently large $N$, there exists a protocol for Distributed Fourier Sampling with inaccuracy $\epsilon$ that communicates at most $\gamma N$ bits. Let $b$ be any positive constant. Then there exists a protocol that communicates at most $2\gamma N$ bits and accepts with probability $\widetilde{p}_{xy}$ on input $(x,y) \in \{\pm1\}^N \times \{\pm1\}^N$ such that:
\[ \E_{(x,y) \sim \xi_{-\sqrt{b/N}}}\left[ \widetilde{p}_{xy} \right] \in \left[\frac{1-\epsilon'}{4N},\frac{1+\epsilon'}{4N}\right], \;\;\;\; \E_{(x,y) \sim \xi_0}\left[ \widetilde{p}_{xy} \right] \in \left[\frac{(b + 1)\beta -\epsilon'}{4N},\frac{(b + 1)\beta+\epsilon'}{4N}\right], \]
\[ \E_{(x,y) \sim \xi_{\sqrt{b/N}}}\left[ \widetilde{p}_{xy} \right] \in \left[\frac{(4b + 1)\beta-\epsilon'}{4N},\frac{(4b + 1)\beta+\epsilon'}{4N}\right], \]
where $\epsilon' = C\eee^{18b} \epsilon$ for some universal constant $C$, and $\beta \in [1-o(1),1]$.
\end{replem}

\begin{proof}
We first show that, for any $N = 2^n$, the existence of a Fourier sampler with the assumed parameters implies the existence of a protocol that communicates $\gamma N$ bits and accepts with probability $\widetilde{p}_{xy}$ for each pair $(x,y)$, such that
\be \label{eq:accuracy} \E_{(x,y) \sim \xi_0}\left[ \left| \widetilde{p}_{xy} - \left(\frac{\braket{x,y}}{N} \right)^2 \right| \right] \le \frac{\epsilon}{N}. \ee
To see this, note that by the definition of Distributed Fourier Sampling and the assumed parameters for the protocol, there exists a protocol that communicates $\gamma N$ bits and samples from some distribution $p'_{xy}$ on bit-strings $s \in \{0,1\}^n$ such that
\[ \epsilon \ge \E_{(x,y) \sim \xi_0}\left[ \sum_{s \in \{0,1\}^n} | p'_{xy}(s) - p_{xy}(s)| \right] = \sum_{s \in \{0,1\}^n} \E_{x,y}[ | p'_{xy}(s) - p_{xy}(s)| ],  \]
where $p_{xy}(s) = \left( \frac{1}{N} \sum_{z \in \{0,1\}^n} (-1)^{s \cdot z} x_z y_z \right)^2$; observe that $p_{xy}(0^n) = (\braket{x,y}/N)^2$. By the above inequality, there is an $s$ such that $\E_{(x,y) \sim \xi_0}[ | p'_{xy}(s) - p_{xy}(s)| ] \le \epsilon / N$. This implies the existence of a protocol which accepts with probability $\widetilde{p}_{xy}$ on input $(x,y)$ such that $\E_{(x,y) \sim \xi_0}[ | \widetilde{p}_{xy} - p_{xy}(s)| ] \le \epsilon / N$; this protocol simply accepts when the outcome is $s$, and rejects otherwise. We can assume that $s = 0^n$ by replacing $x$ with the string $x^{(s)}$ defined by $x^{(s)}_z = (-1)^{s \cdot z} x_z$, as $p_{xy}(s) = p_{x^{(s)}y}(0^n)$. As $(x,y)$ were uniformly distributed, so are $(x^{(s)},y)$.

Alice and Bob will apply their protocol which achieves the bound (\ref{eq:accuracy}) to a pair of inputs $x',y' \in \{\pm1\}^{2N}$ given by a shifted version of their original inputs $(x,y)$. For any $\alpha \in [-1,1]$ and any $\alpha' \in [0,1]$, given a pair of inputs $(x,y) \sim \xi_\alpha$, Alice and Bob can produce a pair of inputs $(x'',y'')$ distributed according to $\xi_{\alpha + \alpha' - \alpha\alpha'}$ by using their shared randomness to do the following: for each $i$, with probability $\alpha'$, replace the pair $(x_i,y_i) \in \{\pm1\}^2$ with a pair $(z_i,z_i)$, where $z_i$ is picked uniformly at random from $\{\pm1\}$. 
Then each of $x''_i$ and $y''_i$ is individually uniform, and
\[ \Pr[x''_i = y''_i] = \alpha' + (1-\alpha') \frac{1+\alpha}{2} = \frac{1}{2} + \frac{\alpha+\alpha' -\alpha\alpha'}{2}. \]
Writing $q = \sqrt{b/N}$, $q' = q / (1+q)$ and then choosing $\alpha \in \{-q,0,q\}$, $\alpha' = q'$, this maps $\xi_{-q} \mapsto \xi_0$, $\xi_0 \mapsto  \xi_{q/(1+q)}$, $\xi_q \mapsto \xi_{2q/(1+q)}$.

Alice and Bob then form a pair of inputs $(x',y')$ by concatenating the inputs $(x'',y'')$ with a pair of inputs in $\{\pm1\}^N$ that differ at exactly $N/2$ positions, and applying the same random permutation of indices to each of their resulting strings. Denote the resulting distribution on pairs~$(x',y')$, when applied to the distribution~$\xi_r$ on~$(x'',y'')$, by~$\xi'_r$. Then clearly
\[ \Pr_{(x',y') \sim \xi'_r}[\braket{x',y'} = \Delta] = \Pr_{(x,y) \sim \xi_r}[\braket{x,y} = \Delta]. \]
By Proposition \ref{fact:techbound}, there is a universal constant $C'$ such that for any $r \in [-0.01,0.01]$,
\[ \Pr_{(x',y')\sim \xi'_r}[\braket{x',y'} = \Delta] \le C' \eee^{2r^2 N} \Pr_{(x',y') \sim \xi^{2N}_0}[\braket{x',y'} = \Delta]. \]
For $r \in \{0,q/(1+q),2q/(1+q)\}$ and sufficiently large $N$, $0 \le r \le 3 \sqrt{b / N}$, where we recall that $q = \sqrt{b/N}$. Write $\operatorname{Err}(x',y') \coloneqq | \widetilde{p}_{x'y'} - (\braket{x',y'}/(2N) )^2 |$ for the error when Alice and Bob apply their protocol to $(x',y')$. Then we have
\beas
\E_{(x',y') \sim \xi'_r}\left[ \operatorname{Err}(x',y') \right] &=& \sum_{\Delta=-N}^N \Pr_{(x',y')\sim \xi'_r}[ \braket{x',y'} = \Delta ]\,\E_{\braket{x',y'} = \Delta}\left[ \operatorname{Err}(x',y')  \right] \\
&\le& \sum_{\Delta=-N}^N C' \eee^{18 b} \Pr_{(x',y')\sim \xi^{2N}_0}[ \braket{x',y'} = \Delta ]\,\E_{\braket{x',y'} = \Delta}\left[ \operatorname{Err}(x',y') \right]\\
&\le&C'\eee^{18 b}\,\E_{(x',y') \sim \xi^{2N}_0}\left[ \operatorname{Err}(x',y') \right]\\
&\le& C'\eee^{18 b} \frac{\epsilon}{N},
 \eeas
where the sums are only taken over even $\Delta$, the inner expectations are uniform over strings $x'$, $y'$ such that $\braket{x',y'} = \Delta$, and the last inequality follows from (\ref{eq:accuracy}). By Proposition~\ref{claim:avgprob}, we have
\[  \E_{(x',y') \sim \xi'_0}\left[ \left( \frac{\braket{x',y'}}{2N} \right)^2 \right] = \frac{1}{4N},\;\;\;\;  \E_{(x',y') \sim \xi'_{q/(1+q)}}\left[ \left( \frac{\braket{x',y'}}{2N} \right)^2\right] = \frac{1}{4}\left(\frac{1}{N} + \left(1- \frac{1}{N}\right) \left(\frac{q}{1+q}\right)^2 \right),\]
\[ \E_{(x',y') \sim \xi'_{2q/(1+q)}}\left[ \left( \frac{\braket{x',y'}}{N} \right)^2\right] = \frac{1}{4} \left( \frac{1}{N} + \left(1- \frac{1}{N}\right) \left(\frac{2q}{1+q}\right)^2 \right).\]
The lemma follows by inserting the definition of $q$ and setting $C = 4C'$.
\end{proof}


\section{Proof of Lemma \ref{lem:tech}: skewed anticoncentration bound}

The proof of Lemma \ref{lem:tech} is analogous to that of Lemma 2.5 in~\cite{chakrabarti12}. Indeed, although it is a generalisation of this result, somewhat remarkably the same proof technique goes through. The starting point is to prove an equivalent result for Gaussian measure.

Let $\gamma^N$ denote the Gaussian distribution on $\R^N$ with density $(2\pi)^{-N/2} \eee^{\|x\|^2/2}$, and let $\Xi_p$ denote the distribution on pairs $(x,y)$ formed by choosing $x,z \sim \gamma^N$, and setting $y = p x + \sqrt{1-p^2}z$. Note that $\Xi_p$ is the ``Gaussian analogue'' of the distribution $\xi_p$ previously introduced.

Let $D(\cdot \| \cdot)$ denote the relative entropy
\[ D(P \| Q) \coloneqq \int P(x) \ln(P(x)/Q(x))\,\mathrm{d}x, \]
where $P$ and $Q$ are probability density functions over $\R$, and write $D_\gamma(X) \coloneqq D(P\|\gamma)$, where $X$ is a random variable with distribution $P$. We will need the following technical claims:

\begin{fact}[Chakrabarti and Regev~{\cite[Claim 3.7]{chakrabarti12}}]
\label{claim:cosh}
For all $\epsilon,\alpha_0 > 0$ there exists a $\delta > 0$ such that for any probability distribution $P$ on $\R$ satisfying $D_\gamma(P) < \delta$, any $z \in \R$, and any $0 < \alpha \le \alpha_0$, we have
\[ \E_{x \sim P}[\cosh(\alpha x + z)] \ge \eee^{\alpha^2/2} - \epsilon. \]
\end{fact}

\begin{thm}[Chakrabarti and Regev~{\cite[Theorem 3.1]{chakrabarti12}}]
\label{thm:closedistrib}
For all $\epsilon,\delta > 0$ and large enough $N$, the following holds. Let $A \subseteq \R^N$ satisfy $\gamma^N(A) \ge \eee^{-\epsilon^2 N}$. Then, for all but an $\eee^{-\delta N / 36}$ measure of unit vectors $y \in \R^N$, the distribution of $\braket{x,y}$ where $x \sim \gamma^N|_A$ is equal to the distribution of $\alpha X + Y$ for some $\alpha \in [1-\delta,1]$ and random variables $X$, $Y$ satisfying $D_\gamma(X\mid Y) \le \epsilon$.
\end{thm}

We will prove the following result:

\begin{thm}
\label{thm:gaussian}
For all $c,\epsilon > 0$ there exists $\delta > 0$ such that for large enough $N$ and $0 \le \eta \le c/\sqrt{N}$, all $A,B \subseteq \R^N$ such that $\gamma^N(A),\gamma^N(B) \ge \eee^{-\delta N}$, and all $s$,
\[ \frac{1}{2} \left( \eee^{s} \Pr_{(x,y) \sim \Xi_{-\eta}}[x \in A \wedge y \in B] + \eee^{-s} \Pr_{(x,y) \sim \Xi_{\eta}}[x \in A \wedge y \in B] \right) \ge (1-\epsilon)\gamma^N(A) \gamma^N(B). \]
\end{thm}

Chakrabarti and Regev's Theorem 3.5 is Theorem \ref{thm:gaussian} with $s=0$~\cite{chakrabarti12}. The proof of Theorem \ref{thm:gaussian} is essentially identical to this special case. However, for completeness and the reader's convenience, we include a full proof. In the proof, $\beta_1,\dots,\beta_7$ are positive constants that have to be sufficiently small.

\begin{proof}[Proof of Theorem \ref{thm:gaussian}]
Let
\[ A' := \{x \in A : (1-\beta_1)N \le \|x\|^2 \le (1+\beta_1)N \} \]
and similarly define $B'$. For sufficiently small $\delta$, by concentration of Gaussian measure, we have
\[ \gamma^N(A') \ge \gamma^N(A) - \beta_2 \eee^{-\delta N} \ge (1-\beta_2) \gamma^N(A) \]
and similarly for $B'$. Write, for any $-1 \le \zeta \le 1$,
\[ p_\zeta \coloneqq \Pr_{(x,y) \sim \Xi_{\zeta}}[x \in A \wedge y \in B]. \] 
Then
\beas
p_\zeta &\ge& \Pr_{(x,y) \sim \Xi_{\zeta}}[x \in A' \wedge y \in B']\\
&=& (2\pi)^{-N/2} (2\pi(1-\eta^2))^{-N/2} \int 1_{A'}(x) 1_{B'}(y) \eee^{-\|x\|^2/2} \eee^{-\|y-\eta x\|^2/(2(1-\eta^2))} \mathrm{d}x\, \mathrm{d}y\\
&=& (1-\eta^2)^{-N/2} \E_{x \sim \gamma^N|_{A'},y \sim \gamma^N|_{B'}} \left[\eee^{-\eta^2 \|x\|^2/(2(1-\eta^2))} \eee^{-\eta^2 \|y\|^2/(2(1-\eta^2))} \eee^{\eta\braket{x,y}/(1-\eta^2)} \right] \gamma^N(A') \gamma^N(B')\\
&\ge& (1-\eta^2)^{-N/2} \eee^{-\eta^2(1+\beta_1)N/(1-\eta^2)}  \E_{x \sim \gamma^N|_{A'},y \sim \gamma^N|_{B'}} \left[\eee^{\eta\braket{x,y}/(1-\eta^2)} \right] \gamma^N(A') \gamma^N(B')
\eeas
%
%
%
%
Applying this inequality with $\zeta = \pm \eta$, we obtain that the quantity we wish to bound, which we can write as $\frac{1}{2}(\eee^s p_{-\eta} + \eee^{-s} p_\eta)$, is lower-bounded by
\be \label{eq:lb} (1-\eta^2)^{-N/2} \eee^{-\eta^2(1+\beta_1)N/(1-\eta^2)} \E_{x \sim \gamma^N|_{A'},y \sim \gamma^N|_{B'}}\left[\frac{1}{2} \eee^s \eee^{-\eta \braket{x,y}/(1-\eta^2)} + \frac{1}{2} \eee^{-s} \eee^{\eta \braket{x,y}/(1-\eta^2)}\right] \gamma^N(A') \gamma^N(B'). \ee
We next observe that
\[ \frac{1}{2} \eee^s \eee^{-\eta \braket{x,y}/(1-\eta^2)} + \frac{1}{2} \eee^{-s} \eee^{\eta \braket{x,y}/(1-\eta^2)} = \cosh(\eta \braket{x,y}/(1-\eta^2) - s). \]
Let $B'' \subseteq B'$ be the set of all $y \in B'$ for which
\[ \E_{x \sim \gamma^N|_{A'}} [\cosh(\eta \braket{x,y}/(1-\eta^2) - s) ] \le (1-\beta_3) \eee^{(\eta/(1-\eta^2))^2(1-\beta_1)N/2}. \] 
We will show that this set is relatively small, $\gamma^N(B'') \le \beta_4 \gamma^N(B')$, implying that
\[ \E_{x \sim \gamma^N|_{A'},y \sim \gamma^N|_{B'}} [\cosh(\eta \braket{x,y}/(1-\eta^2) - s) ] \ge (1-\beta_4)(1-\beta_3) \eee^{(\eta/(1-\eta^2))^2(1-\beta_1)N/2} \]
and hence
\beas (\ref{eq:lb}) &\ge& (1-\eta^2)^{-N/2} \eee^{-\eta^2(1+\beta_1)N/(1-\eta^2)} (1-\beta_4)(1-\beta_3) \eee^{(\eta/(1-\eta^2))^2(1-\beta_1)N/2} \gamma^N(A') \gamma^N(B')\\
&\ge& \eee^{N\eta^2/2} \eee^{-\eta^2(1+\beta_1)N/(1-\eta^2)} (1-\beta_4)(1-\beta_3) \eee^{(\eta/(1-\eta^2))^2(1-\beta_1)N/2} (1-\beta_2)^2 \gamma^N(A) \gamma^N(B)\\
&\ge& (1-\epsilon) \gamma^N(A) \gamma^N(B)
\eeas
for sufficiently small $\beta_1$, $\beta_2$, $\beta_3$, $\beta_4$ and sufficiently large $N$. Assume the contrary for a contradiction, i.e.\ that $\gamma^N(B'') > \beta_4 \gamma^N(B') \ge \beta_4(1-\beta_2)\eee^{-\delta N}$.

Let $r \in [\sqrt{(1-\beta_1)N},\sqrt{(1+\beta_1)N}]$ be such that the Haar measure of points in $B''$ of norm $r$, normalised (and hence restricted to the sphere), is at least $\gamma^N(B'')$; existence of such an $r$ follows from spherical symmetry of the Gaussian distribution and the definition of $B''$. Next, apply Theorem \ref{thm:closedistrib} with $\epsilon$ chosen to be $\beta_5$, $\delta$ chosen to be $\beta_6$, and $A$ chosen to be $A'$. Choosing $\delta$ (in the present theorem) to be small enough, there exists $y \in B''$ such that the distribution of $\braket{x,y}$, where $x \sim \gamma^N|_{A'}$, is given by $\alpha r X + r Y$ for some $\alpha \in [1-\beta_6,1]$ and random variables $X$, $Y$ satisfying $D_\gamma(X\mid Y) \le \beta_5$. This in turn implies that
\[ \Pr_Y[D_\gamma(X|Y) \le \sqrt{\beta_5}] \ge 1 - \sqrt{\beta_5}. \]
%
%
By Proposition \ref{claim:cosh}, for sufficiently small $\beta_7$,
\beas
\E_{x \sim \gamma^N|_{A'}} [\cosh(\eta \braket{x,y}/(1-\eta^2) - s) ] &=& \E[\cosh(\eta (\alpha r X + r Y) /(1-\eta^2) - s) ]\\
&\ge& (1-\sqrt{\beta_5}) (\eee^{(\alpha r \eta/(1-\eta^2))^2/2}-\beta_7) \\
&\ge& (1-\sqrt{\beta_5}) (\eee^{(\eta/(1-\eta^2))^2 (1-\beta_6)^2(1-\beta_1)N/2}-\beta_7) \\
&>& (1-\beta_3) \eee^{(\eta/(1-\eta^2))^2(1-\beta_1)N/2},
\eeas
contradicting the assumption that $y \in B''$.
\end{proof}

We can now prove the corollary we need for the boolean cube (which is Lemma \ref{lem:tech} stated in a slightly more general form).

\begin{cor}
\label{cor:cube}
For all $c,\epsilon > 0$ there exists $\delta > 0$ such that for large enough $N$ and $0 \le p \le c/\sqrt{N}$, all rectangles $R \subseteq \{\pm1\}^N \times \{\pm1\}^N$ such that $\xi_0(R) \ge 2^{-\delta N}$, and all $s$,
\[ \frac{1}{2} \left( \eee^{s} \xi_{-p}(R) + \eee^{-s} \xi_p(R) \right) \ge (1-\epsilon)\xi_0(R). \]
\end{cor}

\begin{proof}
Again, the argument is essentially the same as in~\cite{chakrabarti12}. For $R = A \times B$, let $A' = \{x \in \R^N: \sgn(x) \in A\}$, and similarly for $B'$. Then one can check that, for any $\eta \in [-1,1]$,
\[ \Pr_{(x,y) \sim \Xi_{\eta}}[x \in A' \wedge y \in B'] = \xi_p(R) \]
for $p = 1 - \frac{2}{\pi} \arccos \eta$. In particular, taking $\eta=0$,
\[ \min\{\gamma^N(A'), \gamma^N(B') \} \ge \gamma^N(A') \gamma^N(B') = \Pr_{(x,y) \sim \Xi_0}[x \in A' \wedge y \in B'] = \xi_0(R) \ge 2^{-\delta N}. \]
For $p \le c/\sqrt{N}$ and sufficiently large $N$, $\eta \le c'/\sqrt{N}$ for some constant $c'$, so we can apply Theorem \ref{thm:gaussian} to infer the claim.
\end{proof}


\section{Proofs of further technical propositions}
\label{sec:finalproofs}

In this section we prove some claims about the distributions $\xi_p$, $\xi'_p$ used in the proof of Lemma \ref{lem:prototoaccprob}. Recall that $\xi^N_p$ is the distribution on pairs $(x,y) \in \{\pm1\}^N$ defined in Definition \ref{dfn:xi}. 

\begin{repfact}{claim:avgprob}
\[ \E_{(x,y) \sim \xi^N_p}\left[ \left( \frac{\braket{x,y}}{N} \right)^2 \right] =  \frac{1}{N} + \left(1- \frac{1}{N}\right) p^2. \]
\end{repfact}

\begin{proof}
\beas \E_{(x,y) \sim \xi^N_p}\left[ \left( \frac{\braket{x,y}}{N} \right)^2 \right] &=& \frac{1}{N^2} \E_{(x,y) \sim \xi^N_p}\left[ \left( \sum_i  x_i y_i \right)^2 \right]\\
&=& \frac{1}{N} + \frac{1}{N^2} \sum_{i \neq j} \E_{(x,y) \sim \xi^N_p}[x_i y_i x_j y_j ]\\
&=& \frac{1}{N} + \frac{1}{N^2} \sum_{i \neq j} \E_{(x,y) \sim \xi^N_p}[x_i y_i]  \E_{(x,y) \sim \xi_p}[x_j y_j ]\\
&=& \frac{1}{N} + \frac{N(N-1)}{N^2} (\E_{(x,y) \sim \xi^N_p}[x_1 y_1])^2\\
&=& \frac{1}{N} + \left(1- \frac{1}{N}\right) \left( \frac{1+p}{2} - \frac{1-p}{2} \right)^2\\
&=& \frac{1}{N} + \left(1- \frac{1}{N}\right) p^2.
\eeas
\end{proof}

\begin{repfact}{fact:techbound}
Let $\xi'_p$ be the distribution on pairs $x',y' \in \{\pm1\}^{2N}$ formed by choosing $(x,y) \sim \xi_p$, then setting $x' = \sigma(x,x'')$, $y' = \sigma(y,y'')$, where $x'',y'' \in \{\pm1\}^N$ are fixed strings such that $\braket{x'',y''} = 0$, and $\sigma \in S_{2N}$ is a random permutation of the indices of the strings. Then there exists a universal constant $C$ such that, for sufficiently large $N$ and any $\Delta$, and $p \in [-0.01,0.01]$,
\[ \Pr_{(x',y')\sim \xi'_p}[\braket{x',y'} = \Delta] \le C \eee^{p^2 N} \Pr_{(x',y') \sim \xi^{2N}_0}[\braket{x',y'} = \Delta]. \]
\end{repfact}

\begin{proof}
Assume $\Delta$ is an even integer, otherwise the claim is trivial. By the definition of the distributions $\xi$, $\xi'$, we have
\[ \Pr_{(x',y') \sim \xi'_{p}}[\braket{x',y'} = \Delta] = \left(\frac{1+p}{2} \right)^{(N+\Delta)/2} \left(\frac{1-p}{2} \right)^{(N-\Delta)/2} \binom{N}{\frac{N+\Delta}{2}},  \]
\[ \Pr_{(x',y') \sim \xi^{2N}_0}[\braket{x',y'} = \Delta] = 2^{-2N} \binom{2N}{N + \frac{\Delta}{2}}. \]
Setting $k=(N+\Delta)/2$, the ratio $R$ of these two quantities is thus
\[ R = 2^N \left(1+p \right)^k \left(1-p \right)^{N-k} \frac{\binom{N}{k}}{ \binom{2N}{k+N/2}} \le 2^N \eee^{p(2k-N)} \frac{\binom{N}{k}}{ \binom{2N}{k+N/2}}. \]
We now upper-bound this expression using the binomial coefficient bounds, valid for all integer $\ell \in [1,N-1]$,
\be \label{eq:binombounds} \sqrt{\frac{N}{8\ell(N-\ell)}} \eee^{N h(\ell/N)} \le \binom{N}{\ell} \le \sqrt{\frac{N}{2\pi \ell(N-\ell)}} \eee^{N h(\ell/N)}, \ee
where $h(x) \coloneqq -x \ln x - (1-x) \ln (1-x)$ is the binary entropy measured in nats~\cite{macwilliams83}.

First consider the case $k \in [0,0.1 N] \cup [0.9N,N]$. For these values of $k$, using $k \le N$ and (\ref{eq:binombounds}), there is a universal constant $C$ such that
\[ 2^N \eee^{p(2k-N)} \frac{\binom{N}{k}}{ \binom{2N}{k+N/2}} \le 2^N \eee^{Np} \frac{\binom{N}{0.1N}}{ \binom{2N}{N/2}} \le C 2^N \eee^{Np} \eee^{N h(0.1) - 2N h(1/4) } \le C \eee^{N(p-0.1)}, \]
implying that, for $p \le 0.01$, $R \le 1$ for sufficiently large $N$.

Now consider the range $k \in [0.1N,0.9N]$. Here we will use the inequality, which follows from (\ref{eq:binombounds}), that for some universal constant $C'$,
\[ R \le C' 2^N \eee^{p(2k-N)} \sqrt{\frac{(k+N/2)(3N/2-k)}{k(N-k)}} \eee^{N(h(k/N) - 2h(k/(2N) + 1/4))}. \]
Defining $x = k/N$ and combining terms, we have
\[ R \le C' \sqrt{\frac{(x+1/2)(3/2-x)}{x(1-x)}} \eee^{N(\ln 2 + p(2x-1) + h(x) - 2h(x/2 + 1/4))}. \]
For these values of $x$,
\[ R \le C'' \eee^{N(\ln 2 + p(2x-1) + h(x) - 2h(x/2 + 1/4))} \]
for some new universal constant $C''$, so to prove the bound claimed in the lemma it suffices to show that $f(x) \coloneqq \ln 2 + p(2x-1) + h(x) - 2h(x/2 + 1/4) \le p^2$ for all $x \in [0,1)$. By concavity of $h$, $\ln 2 + h(x) = h(1/2) + h(x) \le 2h(x/2+1/4)$, so $f(x) \le p(2x-1)$, which is negative for $x<1/2$. So, defining $F(x) \coloneqq f(x+1/2) - 2px = \ln 2 + h(x+1/2) - 2h(x/2 + 1/2)$, we want to find an upper bound on $F(x)$ over $x \in [0,1/2)$. We will show that $F(x) \le -x^2$ and hence $f(x+1/2) \le 2px - x^2 \le p^2$, where we take the maximum over $x \in [0,1/2)$.

We have $F(0) = 0$, and
\[ F'(x) \coloneqq \frac{d}{dx} F(x) = \ln \left(\frac{1-2x}{1+2x} \right) - \ln \left(\frac{1-x}{1+x} \right), \]
%
so $F'(0) = 0$ as well. Thus it is sufficient to show that $F''(x) \le -2$ for all $x \in [0,1/2)$ to prove the desired bound. One can calculate that
\[ F''(x) = -\frac{2 + 4 x^2}{1 - 5 x^2 + 4 x^4} \le -2 \]
for $x$ in this range, completing the proof.
\end{proof}


\subsection*{Acknowledgements}

I would like to thank Laura Man\v{c}inska for helpful discussions on the topic of this paper, Tony Short for the ``anti-Holevo'' terminology, Sandu Popescu for pointing out ref.~\cite{galvao03}, and Scott Aaronson for notifying me of ref.~\cite{aaronson16}. I would also like to thank several referees for their very helpful comments on previous versions of this paper. This work was supported by  EPSRC Early Career Fellowship EP/L021005/1, EPSRC grant EP/R043957/1 and the QuantERA ERA-NET Cofund in Quantum Technologies implemented within the European Union's Horizon 2020 Programme (QuantAlgo project).

\appendix


\section{Classical upper bound for the vector in subspace problem}
\label{app:raz}

In this appendix we prove a general upper bound on the amount of classical communication required to solve a bounded-error (and therefore at least as hard) variant of the ``vector in subspace'' problem discussed in Section \ref{sec:intro}. In this version of the problem, Alice gets an $n$-qubit quantum state $\ket{\psi}$ and Bob gets a 2-outcome projective measurement $\{M,I-M\}$. Alice and Bob are promised that either $\braket{\psi|M|\psi} \ge 2/3$ or $\braket{\psi|M|\psi} \le 1/3$; their task is to output 1 in the first case, and 0 in the second.

We will show that there is a classical randomised protocol for this problem that communicates $O(\sqrt{N})$ bits, where as before we set $N=2^n$. This implies that any one-way quantum protocol for computing a partial boolean function on $n$ bits can be simulated by a classical protocol communicating $O(\sqrt{N})$ bits~\cite{kremer95}. This protocol was proposed by Raz~\cite{raz99} (and subsequently also described by Klartag and Regev~\cite{klartag11}), who stated this complexity bound, but no proof has appeared. We first note that we can assume that $\tr M = 2^{n-1}$ without loss of generality; if not, we embed $\ket{\psi}$ and $M$ appropriately in a space of dimension $N' \le 2N$ such that $\tr M = N'/2$, which does not substantially affect the complexity bounds.

We will need the following technical lemma:

\begin{lem}[Corollary of Bennett et al.~\cite{bennett05}]
\label{lem:randomproj}
Let $\ket{\psi}$ be picked from $\C^N$ according to Haar measure on the unit sphere, and let $P$ be the projector onto an $r$-dimensional subspace of $\C^N$. Then, for any $\delta \ge 0$,
\[ \Pr\left[\braket{\psi|P|\psi} \ge (1+\delta) \frac{r}{N}\right] \le \begin{cases} \exp(-r\delta^2/3) & [0 \le \delta \le 1]\\ \exp(-r\delta/3) & [\delta \ge 1] \end{cases} \]
\end{lem}

The protocol proceeds as follows. Alice and Bob use shared randomness to specify $K$ quantum states $\ket{\phi_1}, \dots, \ket{\phi_K}$, each picked independently according to Haar measure, for some $K$ (it will turn out that $K = 2^{\Theta(\sqrt{N})}$ suffices). Then Alice finds the state $\ket{\phi_i}$ such that $|\braket{\phi_i|\psi}|$ is maximised, and sends the identity of this state to Bob using $\lceil \log_2 K \rceil$ classical bits of communication. Bob then outputs 0 if $\braket{\phi_i|M|\phi_i} \le 1/2$, and 1 otherwise.

We first observe that $\ket{\phi_i}$ can be expressed as
\[ \ket{\phi_i} = \epsilon \ket{\psi} + \sqrt{1-\epsilon^2} \ket{\psi^\perp}, \]
where $\epsilon$ is a random variable such that $\epsilon = \braket{\psi|\phi_i}$, which can be taken to be real and positive such that $\epsilon = |\braket{\psi|\phi_i}|$, and $\ket{\psi^\perp}$ is a unit vector distributed uniformly at random in the subspace of states orthogonal to $\ket{\psi}$. This holds because the Haar measure is unitarily invariant, so for each $j$, $\ket{\phi_j}$ can be picked by choosing its inner product with each vector in an arbitrary orthonormal basis according to a complex Gaussian distribution, then normalising the resulting vector. 

Assume that $\braket{\psi|M|\psi} \le 1/3$ (the case $\braket{\psi|M|\psi} \ge 2/3$ is similar, swapping the roles of $M$ and $I-M$). Then the probability that Bob fails to output the correct answer is
\be \label{eq:bound} \Pr_{\ket{\phi_i}}\left[\braket{\phi_i|M|\phi_i} \ge \frac{1}{2}\right]. \ee
We can expand
\beas \braket{\phi_i|M|\phi_i} &=& \left(\epsilon \bra{\psi} + \sqrt{1-\epsilon^2} \bra{\psi^\perp} \right) M \left( \epsilon \ket{\psi} + \sqrt{1-\epsilon^2} \ket{\psi^\perp} \right)\\
&=& \epsilon^2 \braket{\psi|M|\psi} + 2\epsilon\sqrt{1-\epsilon^2} \Re(\braket{\psi^\perp|M|\psi}) + (1-\epsilon^2) \braket{\psi^\perp|M|\psi^\perp},
\eeas
and using $\braket{\psi|M|\psi} \le 1/3$ and a union bound, can upper-bound (\ref{eq:bound}) as
\[ \Pr\left[\braket{\phi_i|M|\phi_i} \ge \frac{1}{2}\right] \le \Pr\left[ 2\epsilon\sqrt{1-\epsilon^2} \Re(\braket{\psi^\perp|M|\psi}) \ge \frac{\epsilon^2}{12} \right] + \Pr\left[ (1-\epsilon^2) \braket{\psi^\perp|M|\psi^\perp} \ge \frac{1}{2} - \frac{5\epsilon^2}{12} \right]. \]
We bound each of the remaining two terms separately, conditioned on the event that $\epsilon \ge \epsilon_0 := c N^{-1/4}$ for some large constant $c$ (we will show that this event occurs with high probability at the end). For the first term, we use
\beas \Pr\left[ 2\epsilon\sqrt{1-\epsilon^2} \Re(\braket{\psi^\perp|M|\psi}) \ge \frac{\epsilon^2}{12} \right] &\le& \Pr\left[ |\braket{\psi^\perp|M|\psi}| \ge \frac{\epsilon_0}{24} \right] \\
&=& \Pr\left[ \frac{|\braket{\psi^\perp|(I-\proj{\psi})M|\psi}|^2}{\|(I-\proj{\psi})M\ket{\psi}\|^2} \ge \left(\frac{\epsilon_0}{24\|(I-\proj{\psi})M\ket{\psi}\|}\right)^2 \right]\\
&\le& \exp(-((\epsilon_0/24)^2(N-1)-1)/3),
\eeas
where in the last step we apply Lemma \ref{lem:randomproj} with $r=1$, $\delta = (\epsilon_0/24)^2(N-1)-1$ to $\ket{\psi^\perp}$ and the projector
\[ P = \frac{(I - \proj{\psi})M \proj{\psi} M (I - \proj{\psi}) }{\tr((I - \proj{\psi})M \proj{\psi} M (I - \proj{\psi})) } \]
and additionally use that $\epsilon_0 \gg N^{-1/2}$. For the second term, we use
\[ \braket{\psi^\perp|M|\psi^\perp} \le \braket{\psi^\perp|\Pi_{((I-\proj{\psi})M(I-\proj{\psi}))}|\psi^\perp}, \]
where $\Pi_X$ denotes the projector onto the support of $X$, and
\[ N/2 = \rank(M) \ge \rank((I-\proj{\psi})M(I-\proj{\psi})) = \rank((I-\proj{\psi})M) \ge \rank(M)-1 = N/2-1, \]
together with Lemma \ref{lem:randomproj} and the bound
\[ \frac{1}{1-\epsilon^2}\left(\frac{1}{2} - \frac{5\epsilon^2}{12}\right) = \frac{6-5\epsilon^2}{12(1-\epsilon^2)} \ge \frac{1}{2} + \frac{\epsilon^2}{12} \ge \frac{1}{2} + \frac{\epsilon_0^2}{12} \]
to obtain
\[ \Pr\left[\braket{\psi^\perp|M|\psi^\perp} \ge \frac{1}{1-\epsilon^2}\left(\frac{1}{2} - \frac{5\epsilon^2}{12} \right) \right] \le \exp( -CN\epsilon_0^4 ) \]
for some universal constant $C$. As $\epsilon_0 = \Omega(N^{-1/4})$, the sum of the two probabilities can be bounded above by an arbitrarily small constant. All that remains is to determine how large $K$ needs to be to achieve $\epsilon \ge \epsilon_0$ with high probability. Recall that $\epsilon$ was the largest absolute inner product between any of the random states $\ket{\phi_j}$ and the fixed state $\ket{\psi}$. For any $0 \le x \le 1$, we have
\[ \Pr_{\ket{\phi_j}}[ |\braket{\phi_j|\psi}|^2 \ge x] = (1-x)^{N-1}. \]
To see this, first note that the distribution obtained by measuring a Haar-random $N$-dimensional state $\ket{\phi_i}$ in an arbitrary orthonormal basis is uniform in the probability simplex~\cite{sykora74}. Geometrically, this corresponds to a standard simplex in $N-1$ spatial dimensions. Truncating this simplex by restricting one coordinate from the range $[0,1]$ to the range $[x,1]$ gives a geometrically similar simplex, whose volume must therefore be the volume of the original simplex multiplied by $(1-x)^{N-1}$. So, for any $0<D<N^{1/2}$,
\[ \Pr_{\ket{\phi_j}}[ |\braket{\phi_j|\psi}|^2 \ge D N^{-1/2}] = (1-D N^{-1/2})^{N-1} \ge (\eee^{-2D N^{-1/2}})^{N-1} \ge \eee^{-2D\sqrt{N}}. \]
It is therefore sufficient to choose $K = 2^{O(\sqrt{N})}$ for there to exist some $i \in \{1,\dots,K\}$ such that $|\braket{\phi_i|\psi}| \ge c N^{-1/4}$ with high probability for large enough $N$, corresponding to $O(\sqrt{N})$ bits of communication being required to specify $\ket{\phi_i}$.


\section{Direct lower bound on the deterministic query complexity of Fourier sampling}
\label{app:query}

In this appendix we give a simple direct proof of Corollary \ref{cor:query} in the special case of deterministic query algorithms; that is, algorithms that choose which bits to query deterministically, then sample from some distribution depending on the bits queried.

\begin{prop}
For sufficiently small constant $\epsilon > 0$, any deterministic classical query algorithm that solves Fourier Sampling on $N$ input bits must make $\Omega(N)$  queries.
\end{prop}

Recall that in this problem we are given query access to a function $h:\{0,1\}^n \rightarrow \{\pm1\}$ (equivalently, to an arbitrary string of $N=2^n$ $\pm1$'s), and are asked to output a sample from any distribution $\widetilde{p}_h$ such that $\| \widetilde{p}_h - p_h \|_1 \le \epsilon$, where
 \[ p_h(s) = \left(\frac{1}{2^n} \sum_{x \in \{0,1\}^n} (-1)^{s \cdot x} h(x) \right)^2. \]

\begin{proof}
Following the same argument as at the start of Lemma \ref{lem:prototoaccprob}, from any classical algorithm which solves Fourier Sampling with $k$ queries, we can obtain an algorithm of the following form. Given oracle access to a bit-string $x \in \{\pm1\}^N$, the algorithm makes $k$ queries to elements of $x$ and accepts with probability $q_x$ such that
\[ \text{Err} := \E_x\left[ \left| q_x - \left(\frac{1}{N} \sum_i x_i \right)^2 \right| \right] \le \frac{\epsilon}{N}, \]
where $x$ is uniformly random. We first argue that we can assume that the algorithm deterministically queries the first $k$ bits of $x$ and its acceptance probability depends only on these. Any deterministic classical algorithm corresponds to a decision tree whose leaves are labelled with acceptance probabilities, and where we can assume without loss of generality that no variable occurs more than once on any path from the root to a leaf. Imagine the first query is to variable $x_j$, where $j \neq 1$, and let $x_{\bar{j}}$ denote the sequence $x_1,\dots,x_{j-1},x_{j+1},\dots,x_n$. Then
\[ \text{Err} = \frac{1}{2} \E_{x_{\bar{j}}} \left[ \left| q^+_{x_{\bar{j}}} - \left(\frac{1}{N} + \frac{1}{N} \sum_{i \neq j} x_i \right)^2 \right| \right] + \frac{1}{2} \E_{x_{\bar{j}}} \left[ \left| q^-_{x_{\bar{j}}} - \left(-\frac{1}{N} + \frac{1}{N} \sum_{i \neq j} x_i \right)^2 \right| \right] \]
for some sequences $q^+$, $q^-$ corresponding to decision trees of depth $k-1$ that do not query $x_j$. By permutation-symmetry of the function $\sum_i x_i$, this quantity would remain unchanged if $x_1$ were relabelled to $x_j$ throughout these decision trees. Following this, the overall tree's first query to $x_j$ can be replaced with a query to $x_1$ without affecting Err.
Performing this transformation recursively, the second query can be assumed to be to variable $x_2$, etc.; ultimately the tree can be replaced with an equivalent one querying variables $x_1,\dots,x_k$ in order.

%
%
Thus, splitting $x$ into a length-$k$ ``queried'' string $y$ and a length-$m$ ``unqueried'' string $z$ such that $k+m=N$, the inaccuracy Err can be written in terms of a sequence of numbers $q_y$ such that
\[ \E_{y,z} \left[ \left| q_y - \left(\frac{1}{N}\left(\sum_i y_i + \sum_j z_j\right)\right)^2 \right| \right] \le \frac{\epsilon}{N}. \]
We will find a lower bound on the left-hand side of this inequality. For convenience, take out a factor of $N^2$, rescaling $q_y$ appropriately. Then minimising this expression is equivalent to minimising
\[ \frac{1}{N^2} \E_{y,z}\!\! \left[ \left| q_y - \left(\sum_i y_i + \sum_j z_j \right)^2 \right| \right] =  \frac{1}{N^2} \E_{y,z}\!\! \left[ \left| q_y- \left(\sum_i y_i\right)^2\!\!\!\!  - \left(\sum_j z_j \right)^2\!\!\!\!  - 2\left(\sum_i y_i\right)\!\!\left(\sum_j z_j\right) \right| \right] . \]
Shifting $q_y$ by $\left(\sum_i y_i\right)^2$ (which we are free to do as $q_y$ is an arbitrary function of $y$), and using the reverse triangle inequality, we can lower-bound this expression by
\[ \frac{1}{N^2} \E_{y,z} \left[ \left| q_y - \left(\sum_j z_j \right)^2 \right| \right] - \frac{2}{N^2} \E_{y,z}\left[\left| \left(\sum_i y_i\right)\left(\sum_j z_j\right) \right| \right]. \]
The random variables $y$ and $z$ are independent and uniformly distributed; and we have $\E_{y}\left[\left| \sum_i y_i\right|\right] = O(\sqrt{k})$, $\E_{z}\left[\left| \sum_i z_i\right|\right] = O(\sqrt{m})$. So the second term is at most $O(\sqrt{km}/N^2)$ in magnitude; and the first term now does not depend on $y$. It remains to bound this term, i.e.\ to lower-bound
\[ \min_q \frac{1}{N^2} \E_z \left[ \left| q - \left(\sum_j z_j \right)^2 \right| \right]. \]
We split into cases. If $q \le m$, this expression is lower-bounded by
\[ \frac{1}{N^2} m \Pr_z\left[\left(\sum_j z_j\right)^2 \ge 2m\right] \ge \frac{Cm}{N^2} \]
for some constant $C$; similarly, if $q \ge m$, we obtain a lower bound of
\[ \frac{1}{N^2} \frac{m}{2} \Pr_z\left[\left(\sum_j z_j\right)^2 \le \frac{m}{2}\right] \ge \frac{C'm}{N^2} \]
for some constant $C'$. Therefore, the whole expression is lower-bounded by $(Cm - D\sqrt{km})/N^2 = (m/N^2)(C-D\sqrt{N/m-1})$ for some universal constants $C$ and $D$. For small enough $\epsilon > 0$, there exists a universal constant $D' < 1$ such that, for all $m \ge D' N$, this is strictly greater than $\epsilon/N$.
\end{proof}

\bibliographystyle{plainnat}
\bibliography{thesis_doi}

\end{document}